\documentclass[nofootinbib,twocolumn]{revtex4-1}
\usepackage[T1]{fontenc}
\usepackage{graphicx}
\usepackage{rotating}
\usepackage{bm}        
\usepackage{amssymb}   
\usepackage{bm}
\usepackage{float}
\usepackage{tikz}
\usepackage{diagbox}
\usepackage{braket}
\usepackage{color} 
\usepackage{colortbl}
\usepackage{amsmath}
\usepackage{csquotes}
\usepackage{amsthm}
\usepackage{bm}
\usepackage{hyperref}
\usepackage{xfrac}

\newtheoremstyle{adaptedtheoremstyle}
  {\topsep}
  {\topsep}
  {}
  {}
  {\itshape}
  {.}
  {.5em}
  {\thmname{#1}\thmnumber{ #2}\thmnote{ (#3)}}
\theoremstyle{adaptedtheoremstyle}

\newtheorem{theorem}{Theorem}
\newtheorem{definition}{Definition}

\newcommand{\BQP}{\textsc{BQP}}
\newcommand{\PromiseBQP}{\textsc{PromiseBQP}}
\newcommand{\Forrelation}{\textsc{Forrelation}}
\newcommand{\kForrelation}{$k$-\textsc{Forrelation}}
\newcommand{\oddkForrelation}{$\textrm{odd-}k$-\textsc{Forrelation}}

\makeatletter

    \def\CT@@do@color{%
      \global\let\CT@do@color\relax
            \@tempdima\wd\z@
            \advance\@tempdima\@tempdimb
            \advance\@tempdima\@tempdimc
    \advance\@tempdimb\tabcolsep
    \advance\@tempdimc\tabcolsep
    \advance\@tempdima2\tabcolsep
            \kern-\@tempdimb
            \leaders\vrule
                    \hskip\@tempdima\@plus  1fill
            \kern-\@tempdimc
            \hskip-\wd\z@ \@plus -1fill }
    \makeatother

\def\k1{k_1}
\def\k2{k_2}
\def\q1{q_1}
\def\q2{q_2}

\def\({\left (}
\def\){\right )}
\def\[{\left [}
\def\]{\right ]}

\newcommand{\beq}{\begin{equation}}
\newcommand{\eeq}{\end{equation}}

\DeclareMathAlphabet\mathbfcal{OMS}{cmsy}{b}{n}

\begin{document}

\date{\today}
\flushbottom \draft
\title{Universal expressiveness of variational quantum classifiers and quantum kernels for support vector machines
}
\author{Jonas J\"{a}ger\textsuperscript{1,2}* (\url{mail@jonas-jaeger.com}) and Roman V. Krems\textsuperscript{1,3}* (\url{rkrems@chem.ubc.ca})} 
\affiliation{
\textsuperscript{1}\,Department of Chemistry, University of British Columbia, Vancouver, B.C. V6T 1Z1, Canada \\
\textsuperscript{2}\,Department of Computer Science, Technische Universit\"{a}t Darmstadt, Darmstadt, Germany \\
\textsuperscript{3}\,Stewart Blusson Quantum Matter Institute, Vancouver, B.C. V6T 1Z4, Canada}

\begin{abstract}
Machine learning is considered to be one of the most promising applications of quantum computing. 
Therefore, the search for quantum advantage of the quantum analogues of machine learning models is a key research goal. 
Here, we show that variational quantum classifiers and support vector machines with quantum kernels can solve 
a classification problem based on the \kForrelation\ problem, which is known to be \PromiseBQP-complete.  Because the \PromiseBQP\ complexity class includes 
all Bounded-Error Quantum Polynomial-Time (\BQP) decision problems, our results imply that there exists a feature map and a quantum kernel that make variational quantum classifiers and quantum kernel support vector machines efficient solvers for any \BQP\ problem.
Hence, this work implies that their feature map and quantum kernel, respectively, can be designed to have a quantum advantage for any classification problem that cannot be classically solved in polynomial time but contrariwise by a quantum computer.\looseness=-1
\end{abstract} 

\maketitle

\section*{Introduction}

Quantum machine learning (QML) has recently emerged as a new research field aiming to take advantage of quantum computing for machine learning (ML) tasks \cite{Biamonte.2007,Ramezani.2020,Li.2020,schuld2021machine}. 
It has been shown that embedding data into gate-based quantum circuits can be used to produce kernels for ML models by quantum measurements \cite{havlivcek2019supervised,schuld2019quantum,schuld2021supervised,mengoni2019kernel,chatterjee2016generalized,glick2021covariant,blank2020quantum}.   
Quantum kernels have been used as kernels of support vector machines (QSVM) for classification \cite{park2020theory,park2020practical,suzuki2020analysis,bartkiewicz2020experimental,wu2021application,haug2021large,liu2021rigorous} and Gaussian process models for regression problems \cite{otten2020quantum,jun}. 
Variational quantum circuits have been used to devise variational quantum classifiers (VQC) \cite{havlivcek2019supervised,Blance.2021,Farhi.2018}.
However, for QML to become a new computational paradigm, it is necessary to prove and demonstrate the computational advantage of ML models based on quantum circuits.

Computational problems are classified in computational complexity theory according to the scaling of 
 time and memory requirements in a computational model with the problem size. 
For example, the classical complexity class  \textsc{P} encompasses all decision problems that are solvable on a deterministic Turing machine in time which scales polynomially with the problem size. 
Analogously, class \textsc{NP} can be defined to encompass problems solvable on a {non}-deterministic Turing machine in polynomial time.
Problems solvable in polynomial time are considered efficient. Hence, decision problems in \textsc{P} are efficiently solvable by classical computers, but it is assumed that this is not the case for problems in NP (\textsc{P} $\neq$ \textsc{NP}).
Problems can further be in special relations to complexity classes. A problem is {complete} relative to a complexity class, if every problem in this class  
can be reduced to this problem under an efficient transformation. 
Another relation is {hardness}. A hard problem relative to a complexity class is at least as difficult to solve as any problem in this class.
Importantly, this implies that hardness is a stronger property than completeness since a hard problem is also complete for a particular class, if it is in this class, but it can be in a hierarchically higher class.\looseness=-1

    \begin{figure}
                \includegraphics[width=0.49\linewidth]{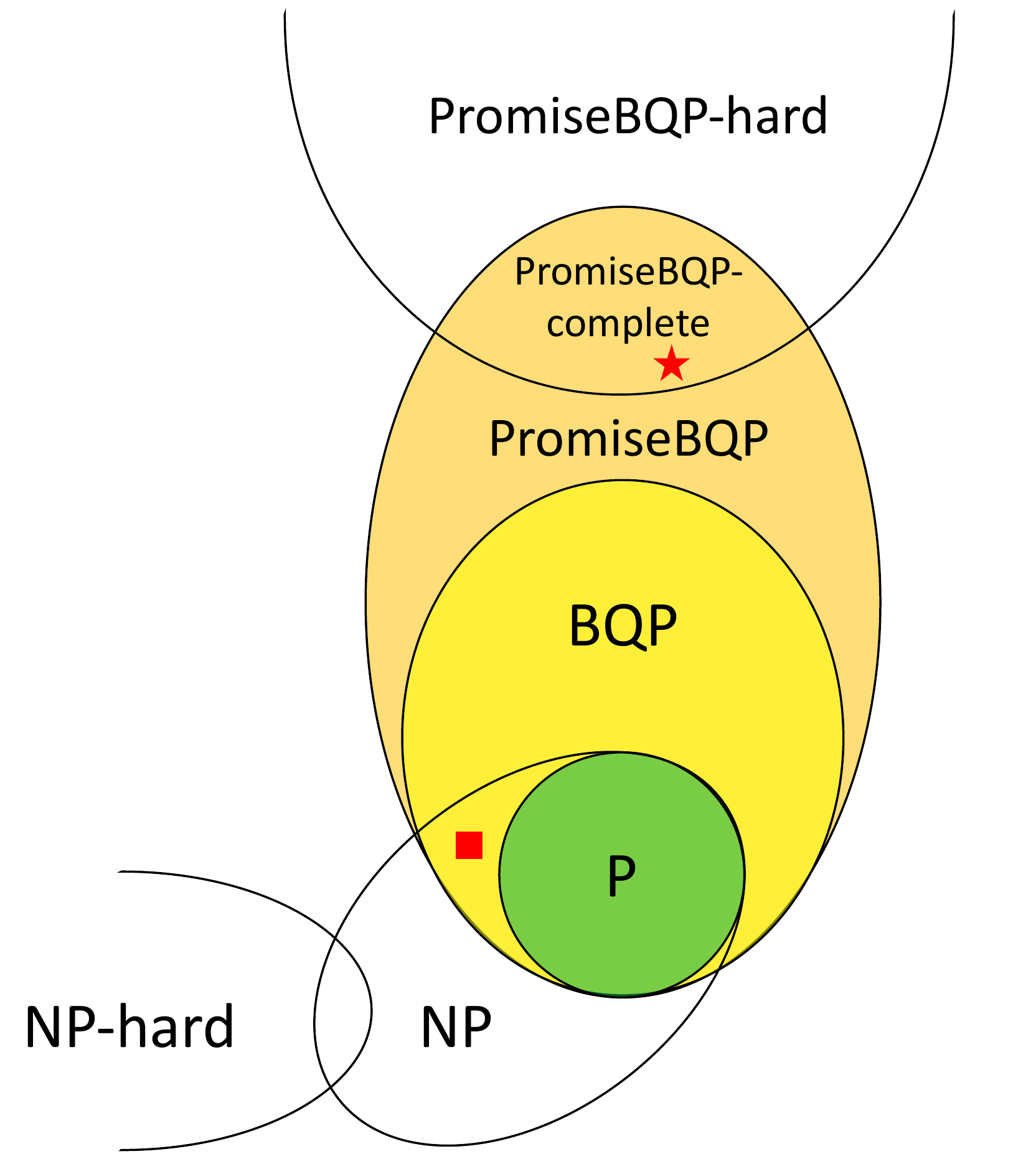}
            \caption{Hierarchy and relations of the complexity classes and problems relevant for this work. This includes the discrete logarithm decision problem $\mathrm{DLP}_{\sfrac{1}{2}}$ ({\color{red}red square}) and (explicit) \kForrelation\ promise problem ({\color{red}red star}). We use the following established, but not yet proven, assumptions: $\mathrm{DLP}_{\sfrac{1}{2}}$ in NP, P $\neq$ NP, P $\neq$ BQP ($\implies$ existence of quantum advantage), NP-complete is outside \BQP, (\textsc{Promise})BQP-complete is outside NP.}
\label{fig:CC}
        \end{figure}

Quantum computing problems are classified by quantum complexity theory \cite{Cleve.2000}. In particular, class 
  \BQP\  -- {bounded-error quantum polynomial time} -- encompasses decision problems solvable in polynomial time by a quantum Turing machine (the uniform family of polynomial-size quantum circuits), with at most 1/3 probability of error. 
While BQP includes P, because all efficient classical computations can be performed deterministically using quantum circuits with polynomial depth, BQP is assumed to also include problems that are not in P. 
This means that BQP-complete problems are not in P.  
Otherwise, BQP would be equal to P and  there would be no quantum advantage to any quantum computing algorithm. 
Thus,  (it is believed that) BQP-complete problems cannot be solved in polynomial time on a classical computer. The hierarchy and relations of complexity classes relevant for this work are shown in Fig. \ref{fig:CC}.

To demonstrate quantum advantage of QSVM, Liu et al. \cite{liu2021rigorous} considered the \textsc{Discrete Logarithm Problem} (DLP). The problem is to find the logarithm $x = \log_g y$ in a multiplicative group of integers modulo prime $p$ (denoted as $\mathbb{Z}_p^*$) for a generator $g$, i.e., such that $g^x \equiv y \pmod{p}$.
DLP is believed, but not rigorously proven, to be unsolvable with polynomial time in the number of bits $n = \lceil \log_2 p \rceil$ on a classical computer. Furthermore, only computing the most significant bit of $x = \log_g y$ for the $\tfrac{1}{2}+\tfrac{1}{\mathrm{poly}(n)}$ fraction of $x\in \mathbb{Z}_p^*$ is as hard as solving DLP \cite{liu2021rigorous,Blum.1984}. This forms a decision problem ($\mathrm{DLP}_{\sfrac{1}{2}}$), presumed to be in NP, which was adopted by Liu et al. \cite{liu2021rigorous} into a classification task to prove separation between QSVM and classical ML classifiers.
Given that $\mathrm{DLP}_{\sfrac{1}{2}}$ is in NP (as shown in Fig. \ref{fig:CC} by the square), it can be argued that $\mathrm{DLP}_{\sfrac{1}{2}}$ cannot be a BQP-complete problem \cite{Raz.2018}. 
Therefore, one cannot generalize the results of Liu et al. \cite{liu2021rigorous} to arbitrary problems in BQP. 

In the present work, we show that VQC and QSVM can solve a problem that is complete in a hierarchically higher class in relation to BQP -- namely, \PromiseBQP. 
As such, our results imply that there exists a quantum kernel or a feature map that  makes VQC and QSVM efficient solvers for any problem with \BQP\ complexity.

\section*{Results}

We use the \kForrelation\ problem that is proven to be \PromiseBQP-complete \cite{Aaronson.2014}.  As defined and described in detail in the Methods section, the \kForrelation\ problem considers $k$ Boolean functions $f_1, \ldots, f_k: \{0,1\}^n \rightarrow \{-1, 1\}$  yielding 
\begin{eqnarray}
         \Phi_{f_{1}, \ldots, f_{k}}  := \frac{1}{2^{(k+1) n / 2}} 
 \sum_{{x_{1}, \ldots, x_{k} \in\{0,1\}^{n}}} f_{1}\left(x_{1}\right)(-1)^{x_{1} \cdot x_{2}} \nonumber
 \\ f_{2}\left(x_{2}\right)(-1)^{x_{2} \cdot x_{3}}  \cdots(-1)^{x_{k-1} \cdot x_{k}} f_{k}\left(x_{k}\right)~~
\end{eqnarray}
with $x \cdot y = \sum_{i=1}^n x_iy_i$.  
We first introduce a classification problem based on the \kForrelation\ promise problem including a compact data encoding scheme. Correctly classifying such a data set requires an algorithm with \PromiseBQP-complete complexity. 

We then show that this classification problem can be solved efficiently and with arbitrary accuracy by both quantum-enhanced classification algorithms: VQC and QSVM, which are reviewed in detail in the Methods section.
Therefore, the resulting classification models solve the \kForrelation\ problem in the \PromiseBQP\ setting and can represent any algorithm to solve all \PromiseBQP\ problems. In other words, we show that these quantum-enhanced classification algorithms are of \PromiseBQP-complete expressive power.

        
\subsection*{\kForrelation\ classification data set} 
We formulate a classification problem with the same complexity as the \kForrelation\ problem. 
Generally, given a promise problem $\Pi = (\Pi_+, \Pi_-)$, one can obtain a data set $\mathcal{D} = \{ \bm{x_i},  y_i \}_{i\in \{1, \ldots, m\}}$ by encoding $m = m_+ + m_-$ instances from $\Pi$ into input vectors $\bm{x_i}$ where the $m_+$ instances sampled from $\Pi_+$ are labeled with class $y_i=+1$ whereas the $m_-$ instances sampled from $\Pi_-$ are labeled with class $y_i=-1$.
Deriving a data set based on the \kForrelation\ problem is not straightforward since the problem instances $\Pi_+ \cup \Pi_-$ consist of $k$-tuples of Boolean functions with $n$-bit inputs for which the description length to encode an instance generally grows exponentially in $n$. Specifically, an arbitrary $n$-bit Boolean function needs $2^n$ bits to encode the evaluation outcome for the $2^n$ possible inputs. Since a \kForrelation\ instance incorporates $k$ such functions, the resulting data set would have dimensionality $k 2^n$.

We use the restriction that each Boolean function $f_i$ depends on at most three input bits as allowed for \kForrelation\ to remain \PromiseBQP-complete as long as the condition is fulfilled that at least one function depends on exactly three bits \cite{Aaronson.2014}. More specifically, each function can be restricted to be either constant $f_i(x) = 1$ or of the form $f_i(x) = (-1)^{C_i(x)}$ where $C_i(x)$ is a product of at most three bits. This enables one to encode a \kForrelation\ instance using up to three indices per function $f_i$ indicating the input bits involved in the product $C_i(x)$ or none indicating the constant function $f_i(x) = 1$. 
We propose an explicit and practically effective multi-hot encoding scheme. Each function $f_i$ can be represented by an $n$-dimensional binary vector where a 1 in the $j$-th component indicates that the $j$-th input bit $x_j$ is incorporated in the product $C_i(x)$. The constant function $f_i(x) = 1$ can be encoded as the zero vector. For example, with $n=3$, the $k=3$ Boolean functions $f_1(x) = (-1)^{x_1x_3}$, $f_2(x) = +1$ and $f_3(x) = (-1)^{x_2}$ would be encoded as $\bm{x} = (1, 0, 1, 0, 0, 0, 0, 1, 0)^\top$.
The resulting encoding of a \kForrelation\ instance and, therefore, the data dimensionality is $kn$, which is linear in $k$ and, since $k=\mathrm{poly}(n)$, polynomial in $n$ instead of exponential in $n$.

Aaronson and Ambainis \cite{Aaronson.2014} established  the quantum algorithm to solve the \kForrelation\ problem with a constant query complexity by encoding the Boolean functions $f_i$ into unitary transformations with $U_{f_i}|x\rangle = f_i(x) |x\rangle\; \forall x\in \{0,1\}^n $, which are diagonal in the computational basis, and applying them successively to the initial state $|0 \rangle^{\otimes n}$ with leading and subsequent Hadamard gates ($H$). The full quantum circuit can be represented as
\begin{eqnarray}
U_F = H^{\otimes n} U_{f_k} H^{\otimes n} \ldots H^{\otimes n} U_{f_1} H^{\otimes n}.
\label{FQC}
\end{eqnarray}
 Note that $f_i (x) = 1$ produces an identity map $U_{f_i} = I$, while $f_i(x) = (-1)^{C(x)}$ with the product $C(x)$ comprising one, two and three bits induces  $Z$, controlled-$Z$ and controlled-controlled-$Z$ gates, respectively, which causes a relative phase-flip conditioned on the values of up to three qubits \cite{Nielsen.2010}.
In the final state $U_F |0^n\rangle$, $\Phi_{f_{1}, \ldots, f_{k}}$ is equal to the amplitude of state $|0\rangle^{\otimes n}$ and can be, therefore, estimated by measurements in the computational basis to decide the \kForrelation\ problem.

We use the feature map $|\Phi(\bm{x})\rangle = U_{\Phi(\bm{x})} |0\rangle^{\otimes n} = U_{F(\bm{x})} |0\rangle^{\otimes n}$ where $U_{F(\bm{x})}$ is 
defined by Eq. (\ref{FQC}) under the $k$ Boolean functions encoded in the data sample $\bm{x}$.
We show that when used for VQC and for kernel construction in QSVM, this feature map leads to classification models
that predict the correct class associated with the \kForrelation\ instance encoded in the data sample $\bm{x}$. This classification 
can be made arbitrarily accurate by increasing the number of measurements estimating  the probability of $|0\rangle^{\otimes n}$ and is perfect given the exact measurement probability.

\subsection*{\kForrelation\ training data}

We now show how to generate positive and negative training samples $\bm{x^+}$ and $\bm{x^-}$ of a classification problem for VQC and QSVM such that the quantum state $|\Phi(\bm{x^\pm})\rangle = U_{F(\bm{x^\pm})}|0\rangle^{\otimes n}$ produced by circuit (\ref{FQC}) in the feature map or quantum kernel corresponds to the positive class sample if all qubits are in state $|0 \rangle$ and the negative class sample if they are in another computational basis state $|z\rangle$ with $0 < z < 2^n$. 
To do this, we use the following theorem, which is proven in the Methods section:

    \begin{theorem}[\oddkForrelation]
    \label{thm:oddkForrelationPromiseBQPHard}~\\
        Explicit $k$-\textsc{Forrelation} remains \PromiseBQP-complete when $k$ is restricted to odd $k \geq 3$.
     \end{theorem}

\noindent


    First, we show how to obtain a positive sample $\bm{x^+}$ such that the initial state is preserved under circuit (\ref{FQC}), i.e., $U_{F(\bm{x^+})}|0\rangle^{\otimes n} = |0\rangle^{\otimes n}$. 
    For odd $k$ Boolean functions,  circuit (\ref{FQC})  includes $k+1$ Hadamard gates, an even number. For all $f_i(x) = +1$, the initial state is preserved since $U_{f_i} = \mathrm{I}$ and the resulting pairs of successive Hadamard gates annihilate. To fulfill the condition that at least one Boolean function must depend on exactly three bits, we choose, without loss of generality, the first and third Boolean functions to be $f_1(x) = f_3(x) = (-1)^{x_ix_jx_l}$. With this choice, 
 \begin{equation}
    \mathrm{H}^{\otimes n}\mathrm{U}_{f_3}\mathrm{H}^{\otimes n}\mathrm{I}\mathrm{H}^{\otimes n}\mathrm{U}_{f_1}\mathrm{H}^{\otimes n} = \mathrm{H}^{\otimes n}\mathrm{U}_{f_3}\mathrm{U}_{f_1}\mathrm{H}^{\otimes n} = \mathrm{I}
 \end{equation}
    since $f_1(x)f_3(x) = (-1)^{2x_ix_jx_l} = 1$. The positive sample $\bm{x^+}$ encoding these functions gives $U_{F(\bm{x^+})}|0\rangle^{\otimes n} = |0\rangle^{\otimes n}$.

    Second, we proceed with generating a negative sample $\bm{x^-}$ for which circuit (\ref{FQC}) maps $|0\rangle^{\otimes n}$ to a different computational basis state, i.e., $U_{F(\bm{x^-})}|0\rangle^{\otimes n} = |z\rangle$ with $0 < z < 2^n$.
    Observe that the unitary $U_{f_i}$ with $f_i(x) = (-1)^{x_j}$ implements a Pauli-$Z$ gate, which resolves to the Pauli-$X$ gate when sandwiched by Hadamard gates $HZH=X$. This flip in qubit $j$ transforms from the initial to another computational basis state $|z\rangle$ with $z_j = 1$. Without loss of generality, we fix $i=1$ and choose a subsequent $f_2(x)$ fulfilling the three-qubit dependence condition for \PromiseBQP-completeness so that all the following $k-1$ Hadamard gates, an even number, pairwise annihilate when the remaining $l>2$ functions are constant $f_l(x) = 1$. Thus, $f_2(x)$ might only cause a global phase-flip on $|z\rangle$, which can be ignored, and preserves the non-zero basis state of qubit $j$ such that $U_{F(\bm{x^-})}|0\rangle^{\otimes n} = |2^{j-1}\rangle \neq |0\rangle$.

\subsection*{Universal expressiveness of VQC}
We first present the proof for VQC. 
The VQC model \cite{havlivcek2019supervised} uses a feature map to encode the input data $\bm{x}$ into an $n$-qubit quantum state $|\Phi(\bm{x})\rangle = U_{\Phi(\bm{x})} |0\rangle^{\otimes n}$ followed by a parameterized quantum circuit $W(\bm{\theta})$. A decision rule, involving an additional bias term $b \in [-1,1]$, enables classification by estimating the binary measurement probability 
\begin{equation}
p_{\pm 1}(\bm{x}) = \langle\Phi(\bm{x})|W^{\dagger}({\bm{\theta}}) M_{\pm1} W({\bm{\theta}})| \Phi(\bm{x})\rangle
\end{equation}
to classify $\bm{x}$ as positive if
    \begin{equation}
    \label{eq:vqc_decision_rule}
        p_{+1}(\bm{x}) > \tfrac{1}{2}(1-b)
    \end{equation}
    or negative otherwise.
    
\begin{proof}
We use proof by reduction where our goal is to find the decision rule 
(\ref{eq:vqc_decision_rule})
to predict class +1 for each instance of the \kForrelation\ problem if and only if it is positive $\bm{x} \in  \Pi_+$.
We start with a data sample $\bm{x}$ that encodes the functions $f_1,\ldots, f_k$ and note that the choice of \kForrelation\ feature map $U_{\Phi(\bm{x})} = U_{F(\bm{x})}$, observable $M_{+1} = |0 \rangle^{\otimes n} \langle 0|^{\otimes n}$ and parameters $\bm{\theta}$ such that $W(\bm{\theta}) = I$ leads to 
\begin{align}
            p_{+1}(\bm{x}) 
= |\langle 0 |^{\otimes n} U_{F(\bm{x})} |0\rangle^{\otimes n}|^2
            = |\Phi_{f_1, \ldots, f_k}|^2.
\end{align}

 For the two possible classes for a data sample $\bm{x}$, two bounds to $b$ can be derived as follows:
\begin{itemize}
    \item If $\bm{x}$ belongs to class $+1$:\\
    $\Phi_{f_1, \ldots, f_k} \geq {3}/{5}$ holds and, therefore, $|\Phi_{f_1, \ldots, f_k}| \geq ({3}/{5})^2 = 9/25$, which, when inserted into the decision rule (\ref{eq:vqc_decision_rule}), yields
    \begin{equation}
    p_{+1}(\bm{x}) \geq \tfrac{9}{25} > \tfrac{1}{2}(1-b).
    \end{equation}
    This only holds if $b$ is chosen to be greater than $-7/25$.
    \item If $\bm{x}$ belongs to class $-1$:\\
    $\Phi_{f_1, \ldots, f_k} \leq {1}/{100}$ holds and, therefore, $|\Phi_{f_1, \ldots, f_k}| \leq ({1}/{100})^2 = 1/10000$. As the decision rule (\ref{eq:vqc_decision_rule}) must be violated, i.e., $p_{+1}(\bm{x}) < (1-b)/2$ for a negative sample $\bm{x}$, a second condition can be derived as
    \begin{equation}
    p_{+1}(\bm{x}) \leq \tfrac{1}{10000} < \tfrac{1}{2}(1-b).
    \end{equation}
    This only holds if $b$ is chosen to be less than $4999/5000$.
\end{itemize}

Thus, the {VQC} decision rule (\ref{eq:vqc_decision_rule}) with the choice of $b \in \left(\tfrac{7}{25}, \tfrac{4999}{5000}\right)$ 
 decides the \kForrelation\ problem.
The existence of values of $\bm{\theta}$ and especially $b$ that allows separation of the two classes was not a priori guaranteed. The demonstration of their existence ensures that VQC has \PromiseBQP-complete expressive power. 
We note again that the transformation from \kForrelation\ to {VQC} is polynomial in time.
\end{proof}

\subsection*{Universal expressiveness of QSVM}

We now present the proof for QSVM.
 The QSVM approach uses a quantum computer to estimate the kernel function 
 \begin{align}
    k(\bm{x_i}, \bm{x_j}) = |\langle \Phi(\bm{x_i}) | \Phi(\bm{x_j}) \rangle|^2 = |\langle 0|^{\otimes n} U^\dagger_{\Phi(\bm{x_i})} U_{\Phi(\bm{x_j})}|0\rangle^{\otimes n}|^2
\end{align}
 which is then used when solving the SVM dual problem \cite{havlivcek2019supervised} classically:
    \begin{align}
    \underset{\bm{\alpha}}{\mathrm{maximize}}\quad &\sum_{i=1}^m \alpha_{i}-\frac{1}{2} \sum_{i=1, j=1}^m \alpha_{i} \alpha_{j} y_{i} y_{j} k(\bm{x_i},  \bm{x_j}) \label{eq:qke_svm_max_probl}\\
    \mathrm{s.t.} \qquad &0 \leq \bm{\alpha} \leq C, \qquad   0 = \sum_{i=1}^m \alpha_{i} y_{i}. \label{eq:qke_svm_constraints}
    \end{align}
    The decision rule for an unseen (test) data sample $\bm{s}$, involving an additional bias term $b \in [-1,1]$, is then
    \begin{equation}
    \label{eq:qke_classification_mapping}
    {m}(\bm{s})=\operatorname{sign}\left(\sum_{i = 1}^m \alpha_{i} y_{i} k(\bm{x_i}, \bm{s})+b\right).
    \end{equation}

\begin{proof}
   We use proof by reduction to show that QSVM can have \PromiseBQP-complete expressive power. 
 The constraints of the dual optimization problem in Eq. (\ref{eq:qke_svm_constraints})
imply that at least two training samples, one from each class, must be provided. 
Therefore, we consider $m=2$ training samples and choose the positive training sample $\bm{x_1} = \bm{x^+}$ with $y_1 = +1$ and the negative training sample $\bm{x_2} = \bm{x^-}$ with $y_2 = -1$ as defined above.
 The equality constraint in Eq. (\ref{eq:qke_svm_constraints}) yields
\begin{equation}
        0 = \alpha_1 y_1 + \alpha_2 y_2 = \alpha_1 - \alpha_2 \iff \alpha_1 = \alpha_2.
\end{equation}
We set $\alpha = \alpha_1 = \alpha_2$, which simplifies the dual optimization problem to one-dimensional optimization constrained on the interval $0 \leq \alpha \leq C$. 
Since $[0, C]$ is a closed and bounded (i.e., compact) interval and the objective function is concave, 
the Weierstraß' extreme value theorem guarantees a maximum on this interval. 
We thus consider $\alpha$ to be the optimal solution, which is guaranteed to be non-negative and can be determined in closed-form in terms of the kernel function evaluated at the two training samples $k(\bm{x_1}, \bm{x_2})$.

As shown earlier, the two training samples produce $U_{F(\bm{x^+})}|0\rangle^{\otimes n} = |0\rangle^{\otimes n}$ and $U_{F(\bm{x^-})}|0\rangle^{\otimes n} = |z\rangle$ with $z \neq 0^n$ when the \kForrelation\ feature map using circuit (\ref{FQC}) is applied.
Under using the \kForrelation\ feature map to construct the kernel, the prediction mapping in Eq. (\ref{eq:qke_classification_mapping}) of QSVM for (test) data sample $\bm{s}$ can now be simplified as
 \begin{align}
            m(\bm{s}) 
            &= \mathrm{sign}\left( \alpha\left(k(\bm{x^+}, \bm{s}) - k(\bm{x^-},\bm{s})\right) + b \right) \\
            &= \mathrm{sign}\left( \alpha\left(|\langle 0^n | U_{F(\bm{s})}|0^n\rangle|^2 - |\langle z | U_{F(\bm{s})}|0^n\rangle|^2 \right) + b \right). \label{eq:qke_proof_classification_mapping_final}
\end{align}
        Here, the two required quantum kernel function estimates correspond to the probabilities to observe the bit-strings $0^n$ and $z$ in the state produced by the \kForrelation\ quantum circuit $U_{F(\bm{s})}|0\rangle^{\otimes n}$ upon measurement in the computational basis.
        
        For the two possible cases $\pm 1$ of a test sample $\bm{s}$, two bounds can be derived for the argument in Eq. (\ref{eq:qke_proof_classification_mapping_final}):
        \begin{itemize}
            \item If $\bm{s}$ belongs to class $+1$:\\
            The measurement probability $|\langle 0^n | U_{F(\bm{s})}|0^n\rangle|^2$ is the absolute squared forrelation quantity $|\Phi_{f_1, \ldots, f_k}|^2$ corresponding to the \kForrelation\ instance encoded in $\bm{s}$, which is $|\langle 0^n | U_{F(\bm{s})}|0^n\rangle|^2 \geq ({3}/{5})^2$ in this case. Since the probabilities have to add up to one, every other $n$-bit bit-string $z \neq 0^n$ can only be observed with a probability of at most $1 - (3/5)^2 = 16/25$, i.e., $|\langle z | U_{F(\bm{s})}|0^n\rangle|^2 \leq 16/25$.
            
            These observations yield a lower bound of
            \begin{equation}
            \alpha( \underbrace{|\langle 0^n | U_{F(\bm{s})}|0^n\rangle|^2}_{\geq (\frac{3}{5})^2} + \underbrace{(-1) |\langle z | U_{F(\bm{s})}|0^n\rangle|^2}_{\geq - 16/25} ) + b \geq   -\tfrac{7}{25}\alpha + b.
            \end{equation}
            Inserting this bound into $m(\bm{s})$, we see that it evaluates to $m(\bm{s}) = +1$ provided $b$ is chosen to be greater than $7\alpha/25$.
            \item If $\bm{s}$ belongs to class $-1$:\\
            Analogously to the previous case, it is known that  $|\langle 0^n | U_{F(\bm{s})}|0^n\rangle|^2 \leq ({1}/{100})^2$ and, therefore, $|\langle z | U_{F(\bm{s})}|0^n\rangle|^2 \geq 1 - ({1}/{100})^2 = 9999/10000$ for any $z \neq 0^n$. Then, the upper bound is
            \begin{equation}
            \alpha( \underbrace{|\langle 0^n | U_{F(\bm{s})}|0^n\rangle|^2}_{\leq (\frac{1}{100})^2} + \underbrace{(-1) |\langle z | U_{F(\bm{s})}|0^n\rangle|^2}_{\leq - (\frac{9999}{10000})} ) + b \leq -\tfrac{4999}{5000}\alpha + b,
            \end{equation}
            and $b$ must be smaller than ${4999}\alpha/{5000}$, which then guarantees that  $m(\bm{s})$ in  Eq. (\ref{eq:qke_proof_classification_mapping_final}) evaluates to -1.
        \end{itemize}
            
        Thus, setting $b \in \left(\tfrac{7}{25}\alpha, \tfrac{4999}{5000}\alpha\right)$ guarantees the correct evaluation of the classification mapping $m(\bm{s})$ for both cases. 
        Again, the existence of $b$ that yields the SVM separating the two classes was not a priori guaranteed. That such an interval exists ensures that QSVM has \PromiseBQP-complete expressive power. 
        \end{proof}

\subsection*{\kForrelation\ fixed ansatz}

    Finally, we show that circuit (\ref{FQC}) used in the feature map or quantum kernel can be implemented using a parameterized quantum circuit with a fixed ansatz, which is typically used in QML. 
    With a single Boolean function $f_i$ in the multi-hot encoding $\bm{x}$,
    the indices $j \in  \lbrace 1, \ldots, n \rbrace$ where $x_j = 1$ determine the target and control qubits of $Z$ gates. 
    To obtain a fixed ansatz, all possible qubit combinations to apply $Z$ gates, controlled-$Z$ gates and controlled-controlled-$Z$ gates in (\ref{FQC}) need to be covered. There are $\binom{n}{1} = n  \in \mathcal{O}(n)$, $\binom{n}{2} = {n(n-1)}/{2} \in \mathcal{O}(n^2)$, $\binom{n}{3} = {n(n-1)(n-2)}/{6} \in \mathcal{O}(n^3)$ possible qubit choices, respectively, due to the gate symmetry \cite{Nielsen.2010}. Instead of a (controlled-) $Z$ gate, a (controlled) rotation about the $Z$ axis $R_Z(\lambda)$ by angle parameter $\lambda$ can be applied as it is equivalent to identity if $\lambda = 0$ and to the (controlled-) $Z$ gate if $\lambda = \pi$. 
    For a controlled rotation gate applied to $J\subseteq \lbrace 1, \ldots, n \rbrace$ qubits, the sample $\bm{x}$ determines $\lambda$ as 
    \begin{equation}
        \lambda = \pi \prod_{j\in J} x_j \prod_{l\in \lbrace 1, \ldots, n\rbrace \setminus J} (1 - x_l)
    \end{equation}
    which gives $\lambda = 0$ in all (controlled) rotation gates except $\lambda = \pi$ for the one that implements $f_i$ encoded in $\bm{x}$.

    For $k$ functions, the fixed ansatz requires $\mathcal{O}(kn^3)$ gates.
    This shows that the expressiveness of VQC and QSVM proven here can be achieved using parameterized quantum circuits with fixed ansatz of polynomial depth since $k = \mathrm{poly}(n)$. This result is important considering that VQC and QSVM are generally implemented using circuits with fixed ansatz \cite{havlivcek2019supervised,schuld2019quantum,schuld2021supervised}.
    However, embedding the data directly through circuit (\ref{FQC}) by applying a single (controlled) $Z$ gate to the respective qubits, which is no longer a fixed ansatz, results in shallower circuits of depth $O(k)$.

\section*{Discussion}
The present work demonstrates that the feature map of VQC and the quantum kernels of QSVM can be used to solve the classification problem with the complexity of the \kForrelation\ problem that has previously 
been proven to be \PromiseBQP-complete. This means that it is possible to design the feature map of VQC and the quantum kernel of QSVM for any classification problem derived from any promise problem in \PromiseBQP. 
Because \PromiseBQP\ includes all decision problems in BQP as a special case, our results imply that it is possible to design the feature map of VQC and the quantum kernel of QSVM that solve any classification problem derived from any decision problem in BQP. 
If BQP $\neq$ \textsc{BPP} (classical bounded error probabilistic polynomial time), as required for exponential speed-up of quantum computing to exist, our results imply that VQC and QSVM must have quantum advantage over classical classifiers.

    According to Havl{\'\i}{\v{c}}ek \cite{havlivcek2019supervised}, every problem that can be solved by {VQC} can also be solved by {QSVM}, but the reverse does not generally hold. This connection is detailed by Schuld \cite{schuld2021supervised} and briefly outlined here. 
    {QSVM} can be seen as {VQC} with an optimal measurement, i.e., $W(\bm{\theta})$ with an optimal ansatz and parameters, since $W(\bm{\theta})$ effectively changes the measurement basis.
    Generally, a fixed ansatz in $W(\bm{\theta})$ requires $\mathcal{O}(2^{2^n})$ degrees of freedom to express arbitrary measurements.
    In QSVM, this reduces to an $m$-dimensional optimization problem as -- in the SVM dual view -- measurements ($\leftrightarrow$ separating hyperplane) become expansions in the training data ($\leftrightarrow$ support vectors). Due to the concavity in Eq. (\ref{eq:qke_svm_max_probl}), this is optimally solved given the kernel values $k(\bm{x_i}, \bm{x_j})$ for all pairs of training data points. Therefore, QSVM is guaranteed to find better or equally good solutions than VQC.
    In the present work, we show that both VQC and QSVM can solve a classification problem based on the \kForrelation\ problem, which implies that VQC and QSVM have an equivalent (universal) expressiveness from a computational complexity theory point of view.

\section*{Methods}  

\subsection*{Quantum-enhanced classification algorithms} 

Two most common, and related, approaches to solving classification problems with quantum computers are VQC and QSVM \cite{havlivcek2019supervised}, schematically depicted in Fig. \ref{fig:QKESVM}. The VQC model first uses a feature map to encode the input data $\bm{x}$ into an $n$-qubit quantum state by a unitary transformation of the initial state $|0\rangle^{\otimes n}$: $|\Phi(\bm{x})\rangle = U_{\Phi(\bm{x})} |0\rangle^{\otimes n}$. Subsequently, a parameterized quantum circuit $W(\bm{\theta})$ transforms the states to enable classification by a quantum measurement. The parameters $\bm{\theta}$ and an additional bias term $b \in [-1,1]$ are learned by classical optimization. A binary measurement probability 
\begin{equation}
p_{\pm 1}(\bm{x}) = \langle\Phi(\bm{x})|W^{\dagger}({\bm{\theta}}) M_{\pm1} W({\bm{\theta}})| \Phi(\bm{x})\rangle
\end{equation}
is estimated to classify $\bm{x}$ as positive if
    \begin{equation}
    \label{eq:vqc_decision_rule-a}
        p_{+1}(\bm{x}) > \tfrac{1}{2}(1-b)
    \end{equation}
    or as negative otherwise
        under choosing two projectors
    \begin{equation}
        M_{\pm}=\frac{1}{2}\left(I \pm \sum_{z=0}^{2^n-1} h_z |z \rangle \langle z |\right)
    \end{equation}
    with arbitrary but fixed coefficients $h_z \in \lbrace -1, 1\rbrace$.
    
    \begin{figure}
            \centering
            \includegraphics[width=\linewidth]{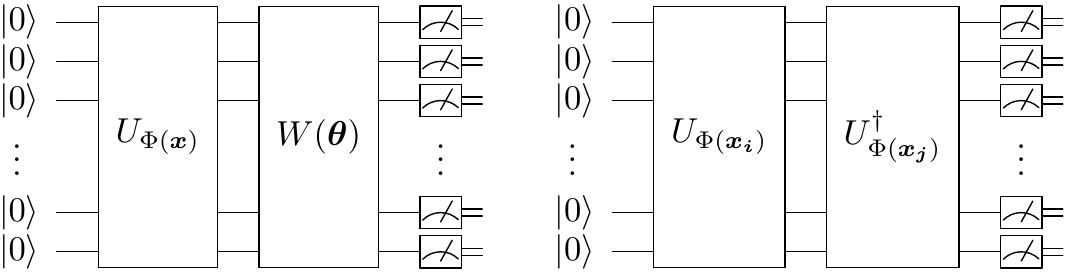}\\\vspace{-0.35cm}
            \null \hspace{1.8cm} \textbf{a} \hfill \textbf{b} \hspace{1.9cm}\null\vspace{-0.2cm}
            \caption{Quantum circuits used in the quantum-enhanced classification algorithms. Diagrams of quantum circuits for \textbf{(a)} variational quantum classifiers ({VQCs}) and \textbf{(b)} quantum kernel support vector machines ({QSVMs}).}
\label{fig:QKESVM}
        \end{figure}
    
 The QSVM approach uses a quantum computer to estimate the kernel function $k(\bm{x_i}, \bm{x_j})$ that is then used in  the dual problem \cite{havlivcek2019supervised}:
    \begin{align}
    \underset{\bm{\alpha}}{\mathrm{maximize}}\quad &\sum_{i=1}^m \alpha_{i}-\frac{1}{2} \sum_{i=1, j=1}^m \alpha_{i} \alpha_{j} y_{i} y_{j} k(\bm{x_i},  \bm{x_j}) \label{eq:qke_svm_max_probl-a}\\
    \mathrm{s.t.} \qquad &0 \leq \bm{\alpha} \leq C, \qquad   0 = \sum_{i=1}^m \alpha_{i} y_{i}. \label{eq:qke_svm_constraints-a}
    \end{align}
    The optimal solution is obtained classically by efficient quadratic optimization and determines the classification mapping of a (test) data sample $\bm{s}$ as
    \begin{equation}
    \label{eq:qke_classification_mapping-a}
    {m}(\bm{s})=\operatorname{sign}\left(\sum_{i = 1}^m \alpha_{i} y_{i} k(\bm{x_i}, \bm{s})+b\right).
    \end{equation}
     Fig. \ref{fig:QKESVM} depicts the quantum circuit to obtain the kernel function
    \begin{align}
        k(\bm{x_i}, \bm{x_j}) = |\langle \Phi(\bm{x_i}) | \Phi(\bm{x_j}) \rangle|^2 = |\langle 0|^{\otimes n} U^\dagger_{\Phi(\bm{x_i})} U_{\Phi(\bm{x_j})}|0\rangle^{\otimes n}|^2
    \end{align}
    as the measurement probability of the $0^n$ bit-string.
        
    \subsection*{Forrelation} 
    
    The complexity classes such as P or BQP are for decision problems with inputs necessarily belonging to `+' or `--' instances. 
  If inputs include a set that corresponds to neither `+' nor `--', the decision problems are generalized to become promise problems \cite{Goldreich.2006}. 
To make decisions, promise problems consider only inputs from the subsets corresponding to the `+/--' instances  (i.e. inputs that are promised to lead to a `+' or '--' decision). 
    
An example of a promise problem is the \Forrelation\ problem introduced in Aaronson \cite{Aaronson.2010}, and refined and extended in Aaronson and Ambainis \cite{Aaronson.2014}.
        This problem considers two Boolean functions $f, g: \{0,1\}^n \rightarrow \{-1, 1\}$ where the domain $\{0,1\}^n$ contains all $2^n$ $n$-bit strings, i.e., the integers from $0$ to $(2^n-1)$ in decimal representation. The quantity 
        \begin{equation}
        \Phi_{f,g} := \frac{1}{2^{3n/2}}\sum_{x,y\in \{0,1\}^n} f(x) (-1)^{x \cdot y} g(y)
        \end{equation}
        with $x \cdot y = \sum_{i=1}^n x_iy_i$
        determines the amount of correlation between $f$ and the Fourier transform of $g$, i.e., the \enquote{forrelation} of $f$ and $g$. Analogously to correlation, one can say that $f$ and $g$ are \enquote{forrelated} once the value $\Phi_{f,g}$ is large or not if it is small.

        The \Forrelation\ problem is solvable with a quantum algorithm \cite{Aaronson.2010} using a single query with error probability of ${2}/{5}$, which can be arbitrarily reduced by increasing the query complexity by a constant factor. Therefore, a quantum algorithm exists that solves the problem with error probability $\leq 1/3$ with a constant number of queries while the query implementing circuit remains polynomial, which makes it a \PromiseBQP\ problem \cite{Aaronson.2014}. As any decision problem is a trivial special case of a more general promise problem, the class of \PromiseBQP\ problems includes BQP entirely, as depicted in Fig.~\ref{fig:CC}.
                
        \subsection*{\kForrelation: a \PromiseBQP-complete extension} 
        
        Aaronson and Ambainis \cite{Aaronson.2014} extended the \Forrelation\ problem to the \kForrelation\ problem. Instead of two Boolean functions, $k$ Boolean functions $f_1, \ldots, f_k: \{0,1\}^n \rightarrow \{-1, 1\}$ are considered and the quantity 
\begin{eqnarray}
         \Phi_{f_{1}, \ldots, f_{k}}  := \frac{1}{2^{(k+1) n / 2}} 
 \sum_{{x_{1}, \ldots, x_{k} \in\{0,1\}^{n}}} f_{1}\left(x_{1}\right)(-1)^{x_{1} \cdot x_{2}} \nonumber
 \\ f_{2}\left(x_{2}\right)(-1)^{x_{2} \cdot x_{3}}  \cdots(-1)^{x_{k-1} \cdot x_{k}} f_{k}\left(x_{k}\right)~~
\end{eqnarray}
with $x \cdot y = \sum_{i=1}^n x_iy_i$
        leads to a promise problem:
        \begin{definition}[\kForrelation]~\\
            The promise problem $\Pi = (\Pi_+, \Pi_-)$ over the space of $k$ Boolean functions $\{0,1\}^n \rightarrow \{-1, 1\}$ with
            \begin{itemize}
                \item $\forall (f_1, \ldots, f_k) \in \Pi_+: \Phi_{f_1, \ldots, f_k} \geq \frac{3}{5}$
                \item $\forall (f_1, \ldots, f_k) \in \Pi_-: |\Phi_{f_1, \ldots, f_k}| \leq \frac{1}{100}$
            \end{itemize}
            is the \kForrelation\ problem. Here, $\Pi_\pm$ are the sets of $\pm$ problem instances with $\Pi_+ \cap \Pi_- = \varnothing$. 
        \end{definition}

\noindent
This definition generally allows the evaluation of the functions $f_1, \ldots, f_k$ by oracle queries. Furthermore, for {explicit} descriptions, which we assume in this work, Aaronson and Ambainis \cite{Aaronson.2014} proved the following theorem:
        \begin{theorem}[\PromiseBQP-completeness]\label{thm:kForrelationPromiseBQPHard}~\\
            If $f_1, \ldots, f_k$ are described explicitly (e.g., by circuits to compute them), and $k = \mathrm{poly}(n)$, then \kForrelation~is \PromiseBQP-complete.
        \end{theorem}
\noindent
and also showed that this still holds when the functions are restricted to depend on at most three input bits of the form $f_i(x) = (-1)^{C_i(x)}$ where $C_i(x)$ is a product of at most 3 input bits, or be chosen constant $f_i(x) = 1$, while at least one $f_i(x)$ must depend on exactly 3 bits in $x$. 
Note the crucial difference: \kForrelation\ (under the stated conditions) is not only a \PromiseBQP\ problem but a \PromiseBQP-{complete} problem.

\subsection*{\oddkForrelation}

Theorem \ref{thm:oddkForrelationPromiseBQPHard} is used for the construction of the data set in the present work. It is restated and proven in the following:

\setcounter{theorem}{0}
    \begin{theorem}[\oddkForrelation]
    \label{thm:oddkForrelationPromiseBQPHard-a}~\\
        Explicit $k$-\textsc{Forrelation} remains \PromiseBQP-complete when $k$ is restricted to odd $k \geq 3$.
     \end{theorem}
     \begin{proof}
 By construction, \oddkForrelation\ is a special case of \kForrelation, which trivially implies that \oddkForrelation\ is in \PromiseBQP. 
        For \PromiseBQP-completeness, it remains to show that \oddkForrelation\ is \PromiseBQP-hard via a proof by reduction: we provide a polynomial mapping from every instance of \kForrelation\ to an instance of \oddkForrelation\ that preserves the forrelation value $\Phi$, which indicates that \oddkForrelation\ is at least as difficult as \kForrelation.
        
If $k$ is odd in an instance of \kForrelation, it is trivially an instance of \oddkForrelation.
If $k$ is even in an instance of \kForrelation, we add $4 \lceil{n/2}\rceil - 1$ Boolean functions resulting in odd $k + 4 \lceil{n/2}\rceil - 1$. The additional functions are chosen such that they are either constant $f(x) = +1$ or of the form $f(x) = (-1)^{x_i x_j}$ with $i,j \in \lbrace 1, \ldots, n \rbrace$, fulfilling the necessary conditions. We show that $\Phi_{f_{1}, \ldots, f_{k}} = \Phi_{f_{1}, \ldots, f_{k + 4 \lceil{n/2}\rceil - 1}}$ as follows.
        
The proof of Theorem 25 in Aaronson and Ambainis \cite{Aaronson.2014} uses a gadget applied to two qubits $i$ and $j$ with $i\neq j$ that converts an even number of $\mathrm{H}^{\otimes 2}$ gates into an odd number. Namely,
        \begin{equation}
        \mathrm{H}^{\otimes 2}\; \mathrm{CZ}\;\mathrm{H}^{\otimes 2} \; \mathrm{CZ}\;\mathrm{H}^{\otimes 2} \; \mathrm{CZ}\;\mathrm{H}^{\otimes 2} \equiv \mathrm{SWAP}\; \mathrm{H}^{\otimes 2}
        \end{equation}
        using three controlled-$Z$ gates ($\mathrm{CZ}$), which implement $f(x) = (-1)^{x_i x_j}$.
        We apply this gadget successively to $\lceil n/2 \rceil$ non-overlapping pairs of qubits to reproduce the final layer of Hadamard gates.
        The gadgets require $3 \lceil{n/2}\rceil$ $\mathrm{CZ}$ gates and $\lceil{n/2}\rceil - 1$ constant functions, so that every fourth of the additional functions produces an identity between two gadgets. In total, an odd number of Boolean functions $f_{k+1}, \ldots f_{k + 4 \lceil{n/2}\rceil - 1}$ is added.
  Obviously, this extends the problem instance from an even to an odd number of Boolean functions, while keeping the circuit equivalent (under SWAP operations) to the original one defined by even $k$ Boolean functions. In other words, the value $\Phi$ is preserved since SWAP operations do not affect the amplitude of $|0\rangle^{\otimes n}$.
  For the pairwise application of the $2$-qubit gadgets in the case of an odd number of qubits $n$, one can introduce an ancilla qubit in $|0\rangle$. The final result remains unaffected as this $(n+1)$-th qubit ends up in $|0\rangle$ and is, therefore, not entangled. 
\end{proof}

\subsection*{Data availability}
Data sharing not applicable to this article as no data sets were generated or analyzed during the current study.

\section*{Acknowledgements}

The work was supported by NSERC of Canada and the German Academic Exchange Service. 

\section*{Author contributions}
J. J. performed the derivations, analysis, and proofs. 
R. V. K. supervised the work. Both authors wrote the article.

\section*{Competing interests}
The authors declare no competing interests.


\end{document}